\documentclass[11pt,oneside,reqno]{amsart}
\usepackage[utf8]{inputenc}
\usepackage{setspace,cite}
\usepackage[margin=1.25in]{geometry}
\usepackage{graphicx}
\usepackage{subcaption}
\usepackage{amsmath}

\title{On the equivalence of linear cyclic and constacyclic codes}
\author{Reza Dastbasteh \and Petr Lison\v{e}k}
\address{Department of Mathematics, Simon Fraser University, Burnaby, BC, Canada}
\email{rdastbas@sfu.ca,  plisonek@sfu.ca}

\usepackage{hyperref}
\hypersetup{
    colorlinks=true,
    linkcolor=blue,
    filecolor=magenta,      
    urlcolor=cyan,
}

\begin{document}

\theoremstyle{plain}
\newtheorem{theorem}{Theorem}[section]
\newtheorem{lemma}[theorem]{Lemma}
\newtheorem{corollary}[theorem]{Corollary}
\newtheorem{proposition}[theorem]{Proposition}
\newtheorem{question}[theorem]{Question}
\theoremstyle{definition}
\newtheorem{notations}[theorem]{Notations}
\newtheorem{notation}[theorem]{Notation}
\newtheorem{remark}[theorem]{Remark}
\newtheorem{remarks}[theorem]{Remarks}
\newtheorem{definition}[theorem]{Definition}
\newtheorem{claim}[theorem]{Claim}
\newtheorem{assumption}[theorem]{Assumption}
\numberwithin{equation}{section}
\newtheorem{example}[theorem]{Example}
\newtheorem{examples}[theorem]{Examples}
\newtheorem{propositionrm}[theorem]{Proposition}

\newcommand{\zar}{{\rm zar}}
\newcommand{\an}{{\rm an}}
\newcommand{\red}{{\rm red}}
\newcommand{\codim}{{\rm codim}}
\newcommand{\rank}{{\rm rank}}
\newcommand{\sing}{{\rm sing}}
\newcommand{\reg}{{\rm reg}}
\newcommand{\Char}{{\rm char}}
\newcommand{\Tr}{{\rm Tr}}
\newcommand{\Nr}{{\rm Nr}}
\newcommand{\res}{{\rm res}}
\newcommand{\tr}{{\rm tr}}
\newcommand{\supp}{{\rm supp}}
\newcommand{\Min}{{\rm Min \ }}
\newcommand{\Max}{{\rm Max \ }}
\newcommand{\Span}{{\rm Span }}
\newcommand{\F}{\mathbb{F}}
\newcommand{\lcm}{{\rm lcm}}
\newcommand{\PP}{\mathcal{P}}
\newcommand{\C}{\mathcal{C}}
\newcommand{\K}{\mathcal{K}}
\newcommand{\X}{\mathcal{X}}
\newcommand{\LL}{\mathcal{L}}
\newcommand{\W}{\mathcal{W}}
\newcommand{\Z}{\mathbb{Z}}
\def\ord{{\rm ord}}
\def\wt{{\rm wt}}

%
%



\begin{abstract}
We introduce new sufficient conditions for permutation and monomial equivalence of linear cyclic codes over various finite fields. 
We recall that monomial equivalence and isometric equivalence
are the same relation for linear codes over finite fields.
A necessary and sufficient condition for the monomial equivalence of linear cyclic codes through a shift map on their defining set is also given.
Moreover, we provide new algebraic criteria for the monomial equivalence of constacyclic codes over $\F_4$. Finally, we prove that if $\gcd(3n,\phi(3n))=1$, then all permutation equivalent constacyclic codes of length $n$ over $\F_4$ are given by the action of multipliers. The results of this work allow us to prune the search algorithm for new linear codes and discover record-breaking linear and quantum codes.
\end{abstract}

\maketitle


\section{Introduction}\label{intro}

Despite the long history and extensive study of linear cyclic codes, several questions regarding their equivalence remain unsolved and presumably are very difficult \cite{Huffman1}. There have been several works toward the classification of equivalent cyclic and constacyclic codes using algebraic properties of these codes; for examples see 
\cite{aydin2017some,aydin2019equivalence,Dobson,guenda,Huffman1, aydin2021, bierbrauer}. 
In the present paper we develop new sufficient conditions for permutation and monomial equivalence of linear cyclic and constacyclic codes. We also resolve two questions raised in the literature regarding the isometric equivalence of linear cyclic codes induced by the action of affine maps on their defining sets. Moreover, we give new sufficient conditions for the isometric equivalence of linear constacyclic codes using the action of affine maps on their defining set. Finally, we prove that two constacyclic codes
over $\F_4$ of an odd length $n$ such that $\gcd(3n,\phi(3n))=1$ are permutation equivalent if and only if there exists a multiplier 
that maps the defining set of one code to the defining set of the other. 

In general, finding linear codes with good parameters is one of the most challenging tasks in algebraic coding theory. A lot of work has been done in the literature to make the computer search for linear codes with good parameters more systematic. However, the computationally challenging obstacles such as 
minimum distance computation, which requires a considerable amount of time, have
slowed down the search process considerably.
Recently, several new linear codes were discovered by designing a more efficient search algorithm for new linear codes using equivalence of linear cyclic and constacyclic codes, see for example 
\cite{akre,aydin2017some,aydin2019equivalence}.
The results of our work can be applied to make the search for new linear codes and also binary quantum codes with good parameters more systematic. In particular, we provide two record-breaking binary quantum codes and one new linear code which are obtained from our computer search. These codes were obtained after pruning the search algorithm for new linear and quantum codes using equivalence 
among linear cyclic and constacyclic codes.

This paper is organized as follows. In Section $\ref{backg}$, we recall the basic notions of equivalence for linear codes. The main contributions of this paper are discussed in Section $\ref{contributions}$.

Section $\ref{cyclic intro}$ reviews several known results from the literature on the equivalence of linear cyclic codes over a general finite field. In Section $\ref{section3.3}$, we introduce novel sufficient conditions for permutation and monomial equivalence of linear cyclic codes over various finite fields. Such conditions help to determine permutation or monomially equivalent cyclic codes, which are not necessarily detectable using the usual methods such as the action of affine maps on defining sets or the generalized multipliers. Next, in Section $\ref{conjecture section}$, we provide a necessary and sufficient condition for the isometric equivalence of linear cyclic codes through a shift map on their defining sets. This resolves a conjecture of Aydin, Lambrinos, and VandenBerg, which was proposed in \cite{aydin2019equivalence}.

Section $\ref{constacyclic1}$ recalls some known results on the equivalence of linear constacyclic codes over $\F_4$. In Section $\ref{constacyclic2}$, we present new results on the equivalence of constacyclic codes over $\F_4$.
Finally, Section $\ref{applications}$ gives examples of two record-breaking binary quantum codes and one new record-breaking linear code over $\F_4$.

 
\section{Background}\label{backg}

Let $\F_q$ be the field of $q$ elements, where $q$ is a prime power.
There are different notions of equivalence for linear codes over $\F_q$. This section studies several such equivalence concepts that preserve the weight distribution when passing from one linear code to another.

\begin{definition}\cite[Section 1.6]{Huffman}\label{def1}
Let $C_1$ and $C_2$ be two linear codes of length $n$ over $\F_q$. Then $C_1$ and $C_2$ are called {\em permutation equivalent} if there exists an $n \times n$ permutation matrix $P$ such that $C_1=C_2P$.
\end{definition}

An $n\times n$ matrix $M$ is called a {\em monomial matrix} over $\F_q$ if $M$ has exactly one non-zero entry from $\F_q$ in each row and each column. Each monomial matrix $M$  over $\F_q$ can be decomposed as $M=PD$, where $P$ is a permutation matrix and $D$ is a non-singular diagonal matrix defined over $\F_q$. 

\begin{definition}\label{equivalence definition cyclic}\cite[Section 1.7]{Huffman}
Let $C_1$ and $C_2$ be two linear codes of length $n$ over $\F_q$. The codes $C_1$ and $C_2$ are called \textit{monomially equivalent} provided that there exists an $n\times n$ monomial matrix $M$ over $\F_q$ such that $C_1=C_2M$. 
\end{definition}

Next, we recall another notion of equivalence for linear codes. Let $C$ and $C'$ be two length $n$ linear codes over $\F_q$ and $\phi$ be an $\F_q$-linear bijection from $C$ to $C'$ preserving the Hamming weight. Then $\phi$ is an isometry of the spaces $C$ and $C'$ equipped with the Hamming distance function $d(x,y)$ as the metric. 
We call such mapping {\em isometry of linear codes} and the codes $C$ and $C'$ {\em isometric equivalent}. An example of isometry of linear codes is the Frobenius automorphism of finite fields. For linear cyclic and constacyclic codes over $\F_q$, the actions of Frobenius maps on codewords can be easily translated in terms of the action of certain multipliers on the defining sets. Multipliers are discussed in Section $\ref{cyclic intro}$. 

Although the isometry of linear codes looks different from the previous notions of equivalence for linear codes, it was proved by MacWilliams in \cite{MacWilliam} that the isometry and the monomial equivalence of binary linear codes are identical concepts in the sense of the next theorem. A generalization of MacWilliams' result over a general finite field is quoted below. 

\begin{theorem}\cite[Corollary 1]{Bogart}\label{isometric=monomial}
Let $C$ and $C'$ be two linear codes over a finite field. Then $C$ and $C'$ are isometric equivalent if and only if these codes are monomial equivalent. 
\end{theorem} 

In spite of the above connection between the monomial equivalence and isometry of linear codes, occasionally, using one of these two definitions can be easier than the other. 

\subsection{Our main contributions}\label{contributions}

There exists significant literature on equivalence of linear codes,
in particular cyclic codes and their generalizations. The reader
may already be familiar with some of these results, yet
we need to recall them thoughout the paper as we are using them
in our proofs. Therefore, in order to seperate the previous results
from the new results, we take the liberty to briefly summarize
our new results in this subsection.


In Section $\ref{section3.3}$ we introduce novel sufficient conditions for the permutation and monomial equivalence of linear cyclic codes. The main results of this section are Theorems $\ref{monomial action}$, $\ref{new permutation}$, and $\ref{F4 permutation}$. 
We also provide a list of code lengths and finite fields containing at least a pair of cyclic codes which are monomially equivalent but not affine equivalent or  permutation equivalent by the action of a generalized multiplier. 

In Section $\ref{conjecture section}$,
Theorem $\ref{affine conjecture}$ resolves a conjecture of \cite{aydin2019equivalence}. Let $C_1$ and $C_2$ be two length $n$ linear cyclic codes over $\F_q$ with the defining sets $A$ and $B$ and $\phi_b$ be a shift map. We prove that $C_1$ and $C_2$ are isometric equivalent through the shift map $\phi_b$ if and only if $\phi_b(A)=B$ and $n$ divides $|A|b(q-1)$.  

In Section $\ref{constacyclic2}$
we first provide a sufficient condition for constacyclic codes over $\F_4$ to have the same parameters using the action of affine maps on their defining sets. A very similar result is stated in Theorem $3.4$ of \cite{akre} 
without proof, and our proof completes this missing piece. 

Let $A$ and $B$ be the defining sets of two length $n$ constacyclic codes over $\F_4$. In Theorem $\ref{shift sufficient}$ we prove that if there exists a shift map $\phi_b$ such that $\phi_b(A)=B$, then $b \equiv 3j \pmod {3n}$ for some $1\le j \le n$ and $n$ divides $3j|A|$.

Theorem $\ref{Palfy consta}$ shows that two constacyclic codes of length $n$ over $\F_4$ such that $\gcd(3n,\phi(3n))=1$ are permutation equivalent if and only if they are permutation equivalent by the action of multipliers. 

In Section $\ref{applications}$ we provide two new binary quantum codes and one new linear code with parameters better than the previous best codes with the same parameters. These codes are obtained by our search algorithm after applying the results of this work.

\section{Equivalence of linear cyclic codes}\label{Linear cyclic codes}


\subsection{A short overview on equivalence of linear cyclic codes}\label{cyclic intro}

Throughout the rest of this section, $n$ always is a positive integer such that $\gcd(n,q)=1$. Let $\alpha$ denote a fixed primitive $n$-th root of unity in $K$ which is a finite field extension of $\F_q$. 

 A linear code $C\subseteq \F_q^n$ is called \textit{cyclic} if for every $c=(c_0,c_1,\cdots,c_{n-1})\in C$, the vector $(c_{n-1},c_0,\cdots,c_{n-2})$ obtained by a cyclic shift of the coordinates of $c$ is also in $C$. It is well-known that there is a one-to-one correspondence between cyclic codes of length $n$ over $\F_q$ and ideals of the ring $\F_q [x]/\langle x^n-1\rangle $, for example see \cite[Section 4.2]{Huffman}. Under this correspondence, each cyclic code $C$ can be uniquely represented by a monic polynomial $g(x)$, where $g(x)$ is the minimal degree generator of the corresponding ideal. The polynomial $g(x)$ is called the {\em generator polynomial} of $C$. Since $\alpha$ is fixed,
any such a linear cyclic code $C$ can be represented by its unique {\em defining set} defined by
$$\{t: 0\le t \le n-1 \ \text{and}\ g(\alpha^t)=0 \}.$$
For each $a\in \Z/n\Z$, the set $Z(a)=\{(aq^j) \mod n: 0\le j \le m-1\}$, where $m$ is the multiplicative order of $q$ modulo $n$, is called 
the \textit{$q$-cyclotomic coset} modulo $n$ containing $a$. The smallest member of a $q$-cyclotomic coset is called its {\em coset leader}. All different $q$-cyclotomic cosets modulo $ n$ partition $\Z/n\Z$. Moreover, the defining set
of a linear cyclic code over $\F_q$ is a union of $q$-cyclotomic cosets.
 
\begin{definition}\label{GEN}
Let $C$ be an $[n,k]$ linear code over $\F_q$ and $H$ be an $(n-k) \times n$ matrix defined over a field extension of $\F_q$. The matrix $H$ is called a {\em generalized parity check matrix} for the code $C$ 
if for each $c\in \F_q^n$
we have $Hc^T=0$ if and only if $ c\in C$.
\end{definition}

The reason we called the matrix $H$ in Definition \ref{GEN} a generalized parity check matrix is that the parity check matrices are defined over the same field as their corresponding linear codes.
Let $\{i_1,i_2,\cdots,i_k\}$ be the defining set for a length $n$ linear cyclic code $C$ over $\F_q$. Then the matrix  
\begin{equation}\label{PC matrix of linear cyclic codes}
H=\begin{bmatrix}
1& \alpha^{i_1} &\alpha^{2i_1} & \cdots & \alpha^{(n-1)i_1} \\
1& \alpha^{i_2} &\alpha^{2i_2} & \cdots & \alpha^{(n-1)i_2} \\
 \vdots&\vdots &\vdots & \cdots &\vdots \\
1& \alpha^{i_k} &\alpha^{2i_k} & \cdots & \alpha^{(n-1)i_k} \\
\end{bmatrix}
\end{equation}
 is a generalized parity check matrix for $C$. 
 
For any integer $a$ such that $\gcd(n,a)=1$, the function $\mu_a$ defined on $\mathbb{Z}/n\Z$ by $\mu_a(x)=(ax) \bmod {n}$ is called a {\em multiplier}. Each multiplier is a permutation of $\Z/n\Z$.  

\begin{theorem}\cite[Theorem 3.2.2]{van}\label{multiplier equivalent cyclic}
Let $C_1$ and $C_2$ be two linear cyclic codes of length $n$ over $\F_q$ with defining sets $A_1$ and $A_2$, respectively. Let $c$ be an integer such that $\gcd(c,n)=1$. If $\mu_c(A_1)=A_2$, then $C_1$ and $C_2$ are permutation equivalent. 
\end{theorem}

The next theorem characterizes all the permutation equivalent linear cyclic codes of certain lengths.

\begin{theorem}\cite[Theorem 1.1]{Huffman1}\label{isoequivalent}
Let $C_1$ and $C_2$ be two linear cyclic codes of length $n$ over $\F_q$ with defining sets $A_{1}$ and $A_{2}$, respectively, and $\gcd(n,\phi(n))=1$. The codes $C_1$ and $C_2$ are permutation equivalent if and only if there exists a multiplier $\mu_a$ such that $\mu_a(A_{1})=A_{2}$. 
\end{theorem}

Let $n=p^m$ and $k\le m$, where $p$ is an odd prime and $k$ and $m\geq 2$ are positive integers. Each element of $\Z/n\Z$ can be uniquely represented as $i+jp^k$, where $0\le i < p^k$ and $0\le j < p^{m-k}$. For any $1\le d < p^k$ such that $\gcd(d,p^k)=1$, the map $M_d:\Z/n\Z \rightarrow \Z/n\Z$ defined by $M_d(i+jp^k)=(id \mod{p^k})+jp^k$ is called a {\em generalized multiplier} of $\Z/n\Z$. 
Let $\mu$ and $M$ be a multiplier and a generalized multiplier defined on $\Z/n\Z$, respectively. The composition map $M\mu$ on $\Z/n\Z$ is defined by $\mu(M(x))$ for each $x \in \Z/n\Z$.

Let $\pi$ be a permutation of $\Z/n\Z$ and
$v=(v_0,\ldots,v_{n-1})\in\F_q^n$.
Define $\pi v$ to be\break
$(v_{\pi^{-1}(0)},\ldots,v_{\pi^{-1}(n-1)})\in\F_q^n$.
The map $v\mapsto \pi v$ is linear over $\F_q$. The matrix $M$ such
that $\pi v=vM$ for each $v$ is called the {\em permutation matrix
corresponding to $\pi$.}

\begin{theorem}\cite[Theorem 3.1]{Huffman1}\label{generalized multipliers}
Let $C_1$ and $C_2$ be two linear cyclic codes of length $p^2$ over $\F_q$, where $p$ is an odd prime and $\gcd(q,p^2)=1$. If $C_1$ and $C_2$ are permutation equivalent, then they are equivalent by the action of permutation matrices corresponding to $\mu$ or $M\mu$ where, $\mu $ is a multiplier and $M$ is a generalized multiplier.
\end{theorem}


We call the map $\phi_b$ on $\mathbb{Z}/n\mathbb{Z}$ defined by $\phi_b(x)=(x+b) \mod {n}$ a {\em shift map}.  The next theorem shows that certain shift maps send the defining set of a linear cyclic code to the defining set of an isometrically equivalent linear cyclic code.

\begin{theorem}\cite[Theorem 3]{aydin2019equivalence}\label{cyclic2}
Let $C_1$ and $C_2$ be two length $n$ linear cyclic codes over $\F_q$ with defining sets $A_1$ and $A_2$, respectively. 
If $b$ is a positive integer such that $n$ divides $ b |A_1|(q-1)$ and $\phi_b(A_1)=A_2$, then $C_1$ and $C_2$ are isometrically equivalent.  
\end{theorem}

The results of Theorems $\ref{multiplier equivalent cyclic}$ and $\ref{cyclic2}$ can be combined to state a more general condition for isometric equivalence of linear cyclic codes.

\begin{corollary}\label{cyclic3}
Let $C_1$ and $C_2$ be two length $n$ linear cyclic codes over $\F_q$ with defining sets $A_1$ and $A_2$, respectively. Let $\theta(x)=(ex+b) \bmod n$ be a map on $\mathbb{Z}/n\mathbb{Z}$, where $e$ and $b$ are positive integers such that $\gcd(e,n)=1$ and $n$ divides $b |A_1|(q-1)$. If $\theta(A_1)=A_2$, then $C_1$ and $C_2$ are isometrically equivalent. 
\end{corollary}

Two isometrically equivalent linear cyclic codes under the action of the map $\theta$ of Corollary \ref{cyclic3} will be called {\em affine equivalent}. 

%
%
%
%
%
%
%
%
%

\subsection{Novel sufficient conditions for equivalence of linear cyclic codes}\label{section3.3}

This section studies sufficient conditions for monomial and permutation equivalence of linear cyclic codes over various finite fields. 
Our conditions are easy to check, and they help to classify all the monomially equivalent cyclic codes of certain lengths. 
Moreover, our new conditions enable us to prove monomial
or permutation equivalence of pairs of codes in cases that
can not be resolved by previously mentioned results. 

Recall that $\alpha$ is a fixed primitive $n$-th root of unity in the field $K$.
We define the vector $v^s\in K^n$ to be 
$v^s=(1 , \alpha^s, \alpha^{2s}, \ldots , \alpha^{(n-1)s})$ for any $0 \le s \le n-1$. We denote the entry in the $j$-th column of a length $n$ vector $v$ by $v_j$ for each $0\le j \le n-1$.

The monomial and permutation equivalence of linear codes can also be defined in terms of the generalized parity check matrices. The following lemma is elementary but we record it for further use in this paper.
  
 \begin{lemma}\label{generalized equivalence}
 Let $C_1$ and $C_2$ be two linear codes of length $n$ over $\F_q$, and let $H$ be a generalized parity check matrix for $C_1$.
Then $C_2=C_1PD$ if and only if $HPD^{-1}$ is a generalized parity check matrix for $C_2$, where $P$ is a permutation matrix and $D$ is a non-singular diagonal matrix defined over $\F_q$.
 \end{lemma}
 
 \begin{remark}
 In Sections $\ref{S:n1}$, $\ref{S:n1}$, and $\ref{S:n1}$, we provide some intermediate lemmas before stating our main result. There might be easier proofs for the equivalence of codes in these lemmas by applying the results of Section \ref{cyclic intro}. However, our main goal is to prove these lemmas using the action of specific permutation or monomial matrices. 
After proving our main results, namely Theorems $\ref{monomial action}$, $\ref{new permutation}$, and $\ref{F4 permutation}$, 
we provide evidence
showing that they are not of the 
types of results discussed in Section \ref{cyclic intro}.
 \end{remark}
\subsubsection{New sufficient condition for monomial equivalence of linear cyclic codes}\label{S:n1}

Let $n$ be a positive integer divisible by $8$ and $\F_q$ be a finite field of odd characteristic. We define the permutation $\sigma$ on $\Z/n\Z$ by
$$\sigma(i)=\begin{cases} i & \text{if}\ i \equiv 0 \ \text{or} \ 1 \pmod 4 \\ (i+\frac{n}{2}) \mod n & \text{otherwise.} \end{cases}$$ 
Let $P_\sigma$ be a permutation matrix corresponding to the action of $\sigma$ and $D$ be a diagonal matrix defined by 
$$D_{(i,i)}=\begin{cases} -1 & \text{if}\ i \equiv 1 \ \text{or} \ 2 \pmod 4 \\1 & \text{otherwise} \end{cases}$$ for each $0\le i \le n-1$. 
Let $\{s_i: 0 \le i \le n-1\}$ be the standard basis of $\F_q^n$. Then 
$$s_i(P_\sigma D)=\begin{cases}
 s_i & \text{if}\ i \equiv 0 \pmod 4\\
 -s_i & \text{if}\ i \equiv 1 \pmod 4\\
  -s_{(i+\frac{n}{2})} & \text{if}\ i \equiv 2 \pmod 4\\
  s_{(i+\frac{n}{2})} & \text{if}\ i \equiv 3 \pmod{4}.
\end{cases}$$
Since the generalized multipliers are not defined over $\Z/ n\Z$ for an even integer $n$, the action of $P_\sigma D$ cannot be of a generalized multiplier type. Later we also show that if two cyclic codes are monomially equivalent under the action of $P_\sigma D$, then they are not necessarily affine equivalent.  
 Before stating our main result, we have two intermediate lemmas.
Recall that $Z(a)$
denotes the $q$-cyclotomic coset modulo $n$ containing $a$. 

\begin{lemma}\label{monomial 1}
Suppose that $q$ is odd.
Let $n=8k$ for some positive integer $k$ and $0\le a \le n-1$ be an odd integer. Then linear cyclic codes over $\F_q$ of length $n$ with the defining sets $Z(a)$ and $Z(\frac{n}{2}+a)$ 
are monomially equivalent under the action of $P_\sigma D$.
\end{lemma}

\begin{proof}
Let $C_1$ and $C_2$ be the linear cyclic codes of length $n$ over $\F_q$ with the defining sets $Z(a)$ and $Z(\frac{n}{2}+a)$, respectively.
First, note that since $q$ is odd 
we have
$q(\frac{n}{2}+a)\equiv \frac{n}{2}+qa \pmod n$. Thus there is a one-to-one correspondence between the elements of $Z(a)$ and $Z(\frac{n}{2}+a)$ given by the shift map $\phi_{\frac{n}{2}}$. Moreover, both $Z(a)$ and $Z(\frac{n}{2}+a)$ consist of only odd values.

Let $H_1$  be a generalized parity check matrix for $C_1$  in the form of $(\ref{PC matrix of linear cyclic codes})$ and  $b$ be an arbitrary element of $Z(a)$. The vector $v^b$ is a row of $H_1$ and we show that $v^bP_\sigma D=v^{b'}$, where $b'=(\frac{n}{2}+b) \mod n$ is an element of $Z(\frac{n}{2}+a)$. 
This shows that $H_1P_\sigma D$ is a generalized parity check matrix for $C_2$.
In our computations, we use the fact that $\alpha^{\frac{n}{2}}=-1$.

Let $0 \le i \le n-1$. If $i\equiv 1 \pmod 4$, then $(v^bP_\sigma D)_i=\alpha^{ib+\frac{n}{2}}=\alpha^{i(\frac{n}{2}+b)}=\alpha^{ib'}$. If $i\equiv 3 \pmod 4$, then $(v^bP_\sigma D)_i=\alpha^{(i+\frac{n}{2})b}=\alpha^{i(\frac{n}{2}+b)}=\alpha^{ib'}$. 
If $i\equiv 0 \pmod 4$, then $(v^bP_\sigma D)_i=\alpha^{ib}=\alpha^{i(\frac{n}{2}+b)}=\alpha^{ib'}$. If $i \equiv 2\pmod 4$, then $(v^bP_\sigma D)_i=\alpha^{(i+\frac{n}{2})b+\frac{n}{2}}=\alpha^{ib}=\alpha^{i(\frac{n}{2}+b)}=\alpha^{ib'}$. 
Hence, for any $i$, we get $(v^bP_\sigma D)_i=(v^{b'})_i$.

Thus $H_1P_\sigma D$ is a generalized parity check matrix for $C_2$. Therefore, by Lemma $\ref{generalized equivalence}$, the codes $C_1$ and $C_2$ are monomially equivalent under the action of $P_\sigma D$.
\end{proof}

Since $8 \mid n$, the $q$-cyclotomic cosets of $0$ and $\frac{n}{2}$ are both singletons. Let $q\equiv 1 \pmod 4$. Then the sets $\{\frac{n}{4}\}$ and $\{\frac{3n}{4}\}$ are two other singleton $q$-cyclotomic cosets. If $q\equiv 3 \pmod 4$, then the set $\{\frac{n}{4}, \frac{3n}{4}\}$ is a $q$-cyclotomic coset. 

\begin{lemma}\label{monomial 2}
Let $n=8k$ for some positive integer $k$ and $A_1=\{0,\frac{n}{2}\}$ and $A_2=\{\frac{n}{4},\frac{3n}{4}\}$ be the defining sets of length $n$ linear cyclic codes $C_1$ and $C_2$ over $\F_q$, respectively. Then $C_1$ and $C_2$ are monomially equivalent under the action of $P_\sigma D$. 
\end{lemma}

\begin{proof}
Since $\alpha^{\frac{n}{2}}=-1$, the code $C_1$ has a parity check matrix in the form 
$$H_1=\begin{bmatrix}
1&1&1&1& \cdots&1 \\
1&-1&1& -1&\cdots&-1\\
\end{bmatrix}.$$
Let 
$$H_2=\begin{bmatrix}
1&\alpha^{\frac{n}{4}}&\alpha^{\frac{2n}{4}}&\alpha^{\frac{3n}{4}}&1& \cdots&\alpha^{\frac{3n}{4}} \\
1&\alpha^{\frac{3n}{4}}&\alpha^{\frac{2n}{4}}&\alpha^{\frac{n}{4}}&1& \cdots&\alpha^{\frac{n}{4}} \\
\end{bmatrix}$$
be a generalized parity check matrix of $C_2$. Next we show that $H_2P_\sigma D$ and $H_1$ generate the same row space over $K$. First, a straightforward computation shows that

$$H_2P_\sigma D=\begin{bmatrix}
1&\alpha^{\frac{3n}{4}}&1&\alpha^{\frac{3n}{4}}&1& \cdots&\alpha^{\frac{3n}{4}} \\
1&\alpha^{\frac{n}{4}}&1&\alpha^{\frac{n}{4}}&1& \cdots&\alpha^{\frac{n}{4}} \\
\end{bmatrix}=\begin{bmatrix}
1&-\alpha^{\frac{n}{4}}&1&-\alpha^{\frac{n}{4}}&1& \cdots&-\alpha^{\frac{n}{4}} \\
1&\alpha^{\frac{n}{4}}&1&\alpha^{\frac{n}{4}}&1& \cdots&\alpha^{\frac{n}{4}} \\
\end{bmatrix}.$$
Next by adding and subtracting the rows of $H_2P_\sigma D$ we find a basis for the row space of $H_2P_\sigma D$ in the form $B=\{(1,0,1,0,\ldots,0),(0,1,0,1,\ldots,1)\}$. One can easily see that the set $B$ is also a basis for the row space of $H_1$. 
Thus $H_2P_\sigma D$ is also a generalized parity check matrix for $C_1$. Now Lemma $\ref{generalized equivalence}$ implies that the codes $C_1$ and $C_2$ are monomially equivalent under the action of $P_\sigma D$.
\end{proof}

Next, we combine the results of Lemmas $\ref{monomial 1}$ and $\ref{monomial 2}$ and state our main result of this section.

\begin{theorem}\label{monomial action}
Let $\F_q$ be a finite field of odd characteristic and $n=8k$ for some positive integer $k$ such that $\gcd(k,q)=1$.  Let $A$ be a union of $q$-cyclotomic cosets modulo $n$ with odd coset leaders. Then the length $n$ linear cyclic codes with the defining sets $A_1=A\cup \{\frac{n}{4},\frac{3n}{4}\} $ and $A_2=\{(a+\frac{n}{2}) \mod n:a \in A\} \cup \{0,\frac{n}{2}\} $ over $\F_q$ are monomially equivalent under the action of $P_\sigma D$. 
\end{theorem}

\begin{proof}
Let $C_1$ and $C_2$ be the length $n$ linear cyclic codes with the defining set $A_1$ and $A_2$ over $\F_q$, respectively.
Let $H_{1}$ be a generalized parity check matrix for $C_1$, in the form of $(\ref{PC matrix of linear cyclic codes})$. Then Lemmas $\ref{monomial 1}$ and $\ref{monomial 2}$ imply that $H_1P_\sigma D$ is a generalized parity check matrix for $C_2$. Thus by Lemma $\ref{generalized equivalence}$ the codes $C_1$ and $C_2$ are monomially equivalent.
\end{proof}

Now we present some applications of Theorem $\ref{monomial action}$. First
let us note that if two linear cyclic codes are monomially equivalent by Theorem $\ref{monomial action}$, then they are not necessarily affine equivalent or permutation equivalent by the action of a generalized multiplier. 

\begin{example}\label{counter example}
Let $A_1=\{ 0, 1, 3, 4 \}$ and $A_2=\{ 2, 5, 6, 7 \}$ be the defining sets of linear cyclic codes $C_1$ and $C_2$  over $\F_3$ of length $8$. One can easily verify that there is no bijective affine map between $A_1$ and $A_2$. Hence $C_1$ and $C_2$ are not affine equivalent. 
Note also that $A_1=\{1,3\}\cup \{0,4\}$ and $A_2=\{5,7\} \cup \{2,6\}$ satisfy the conditions of Theorem $\ref{monomial action}$ and therefore $C_1$ and $C_2$ are monomially equivalent over $\F_3$. Moreover, generalized multipliers are only defined on integers modulo an odd prime power. Hence $C_1$ and $C_2$ are not permutation equivalent by the action of generalized multipliers. 

Finally, our computation in Magma \cite{magma} shows that all the monomially equivalent cyclic codes of length $8$ over $\F_3$, $\F_7$, and $\F_{11}$ are either affine equivalent, monomially equivalent by the action of $P_\sigma D$, or a combination of both.
 \end{example}

Another application of Theorem $\ref{monomial action}$ is to check whether two linear cyclic codes are isodual.
A linear code is called {\em isodual} if it is monomially equivalent to its Euclidean dual. Isodual codes could also be defined similarly in terms of other inner products; however, in this work, we only consider Euclidean isodual codes. 
As Example 4.3 of \cite{fan} shows, solely applying affine bijections on the defining sets of linear cyclic codes fails to detect the existence of isodual cyclic codes of length $8$ over $\F_3$. However, by applying the result of Theorem $\ref{monomial action}$, we are able to show the existence of isodual cyclic codes of length $8$ over $\F_3$.

\begin{example}
The $3$-cyclotomic cosets modulo $8$ are $\{0\}$, $\{1,3\}$, $\{2,6\}$, $\{4\}$, and $\{5,7\}$. Let $C$ be a linear cyclic code over $\F_3$ of length $8$ with the defining set $A=\{0,1,3,4\}$. Its Euclidean dual $C^\bot$ has the defining set $A'=(\Z/n\Z)\setminus (-A)=\{1,2,3,6\}$.
Let $b=4$. Then $b$ satisfies the conditions of Theorem $\ref{cyclic2}$ and $\phi_4(A')=\{2,5,6,7\}$. Thus $C^\bot$ is isometrically equivalent to the linear cyclic code $D$ with the defining set $\{2,5,6,7\}$ over $\F_3$. Moreover, as we showed in Example $\ref{counter example}$, the cyclic codes $C$ and $D$ are monomially equivalent. 
Therefore, $C$ and $C^\bot$ are monomially equivalent. This makes $C$ and $C^\bot$ a pair of isodual codes over $\F_3$ of length $8$.
\end{example}

\subsubsection{New sufficient conditions for permutation equivalence of linear cyclic codes}\label{S:n2}

Let $n$ be a positive integer divisible by $8$. We define the permutation $\gamma$ on $\Z/n\Z$ by
\[
\gamma(i)=
\begin{cases} 
i & \text{if}\ i \ \text{is even} \\ 
(i-2) \bmod n &\text{if}\ i \ \text{is odd.}\\
\end{cases}
\] 
 We denote the permutation matrix corresponding to $\gamma$ by $P_\gamma$.
Let $\{s_i: 0 \le i \le n-1\}$ be the standard basis of $\F_q^n$. Then 
$$s_iP_\gamma=\begin{cases}
 s_i & \text{if}\ i \ \text{is even}\\
s_{i-2} & \text{if}\ i \ \text{is odd.}\\
\end{cases}$$
Since the generalized multipliers are not defined over $\Z/ n\Z$ for an even integer $n$, the action of $P_\gamma$ cannot be of a generalized multiplier type.
Our main result of this section relies on the following intermediate lemmas. 
Recall that $\alpha$ is a fixed primitive $n$-th root of unity in the field $K$ and $v^s\in K^n$ is defined by
$v^s=(1 , \alpha^s, \alpha^{2s}, \ldots , \alpha^{(n-1)s})$ for any $0 \le s \le n-1$.

\begin{lemma}\label{trivial action1}
Let $\F_q$ be a finite field of odd characteristic and $n=8k$ for some positive integer $k$ such that $\gcd(k,q)=1$. 
\begin{enumerate}
\item Let $A=\{0\}$ or $A=\{\frac{n}{2}\}$ and $C$ be the length $n$ linear cyclic code over $\F_q$ with the defining set $A$. Then $C=CP_\gamma$.
\item Let $q\equiv 1 \pmod 4$ and $C_1$ and $C_2$ be the length $n$ linear cyclic codes over $\F_q$ with the defining sets $\{\frac{n}{4}\}$ and $\{\frac{3n}{4}\}$, respectively. Then $C_2=C_1P_\gamma$.
\end{enumerate}
\end{lemma}

\begin{proof}
If $A=\{0\}$, then the all-ones vector $v^0=(1,1,\ldots,1)$ is a parity check matrix for $C$. If $A=\{\frac{n}{2}\}$, then  $v^{\frac{n}{2}}=(1, -1,1, -1, \ldots, -1 )$ is a parity check matrix for $C$. A straightforward computation shows that $v^0P_\gamma=v^0$ and $v^{\frac{n}{2}}P_\gamma=v^{\frac{n}{2}}$. Hence $C=CP_\gamma$.

If $q\equiv 1 \pmod 4$, then $\{\frac{n}{4}\}$ and $\{\frac{3n}{4}\}$ are singleton $q$-cyclotomic cosets modulo $n$. 
The code $C_1$ has a parity check matrix in the form $v^{{\frac{n}{4}}}=(1,\alpha^{{\frac{n}{4}}}, \alpha^{{\frac{n}{2}}},  \alpha^{{\frac{3n}{4}}}, 1 , \alpha^{{\frac{n}{4}}}, \ldots , \alpha^{{\frac{3n}{4}}} )$. Moreover, $v^{{\frac{n}{4}}}P_\gamma=v^{{\frac{3n}{4}}}$. 
Therefore, by Lemma $\ref{generalized equivalence}$, we have $C_2=C_1P_\gamma$.
\end{proof}

\begin{lemma}\label{trivial action2}
Let $\F_q$ be a finite field of odd characteristic and $n=8k$ for some positive integer $k$ such that $\gcd(k,q)=1$. 
 Let $C$ be a length $n$ linear cyclic code over $\F_q$ with the defining set $A=Z(a)\cup Z(a+\frac{n}{2})$ for some $a \in \Z/n\Z$. Then $C=CP_\gamma$.
\end{lemma}

\begin{proof}
If $a=0$ or $a=\frac{n}{2}$, then the proof follows from Lemma $\ref{trivial action1}$. So we assume that $a \neq 0,\frac{n}{2}$. 
Note also that since $q$ is odd 
we have $\phi_{\frac{n}{2}}(Z(a))=Z(a+\frac{n}{2})$. So $Z(a)$ and $Z(a+\frac{n}{2})$ have the same  size.  Let $Z(a)=\{a_1,a_2,\ldots,a_r\}$. 
Next we show that the sets $\{v^{a_i},v^{(a_i+\frac{n}{2})}\}$ and $\{v^{a_i}P_\gamma ,v^{(a_i+\frac{n}{2})}P_\gamma \}$ generate the same vector space over $K$ for each $1\le i \le r$. This shows that if $H$ is a generalized parity check matrix for $C$, in the form of $(\ref{PC matrix of linear cyclic codes})$, then $HP_\gamma$ and $H$ generate the same row space over $K$. Hence Lemma $\ref{generalized equivalence}$ implies that $C=CP_\gamma$. Note that $\alpha^{\frac{n}{2}}=-1$, and
$$v^{a_i}=(1 , \alpha^{a_i} , \alpha^{2a_i} , \alpha^{3a_i} , \ldots , \alpha^{(n-1)a_i} ) $$
 and 
 $$v^{(a_i+\frac{n}{2})}=(1 , \alpha^{a_i+\frac{n}{2}} , \alpha^{2a_i} , \alpha^{3a_i+\frac{n}{2}} , \ldots, \alpha^{(n-1)a_i+\frac{n}{2}} ). $$ 
By adding and subtracting the vectors $v^{a_i}$ and $v^{(a_i+\frac{n}{2})}$ we get  
$$v^{a_i}+v^{(a_i+\frac{n}{2})}=(2 , 0, 2\alpha^{2a_i}, 0 , 2\alpha^{4a_i} , \ldots , 0 )$$
 and 
 $$v^{a_i}-v^{(a_i+\frac{n}{2})}=(0, 2\alpha^{a_i}, 0, 2\alpha^{3a_i}, 0 , 2\alpha^{5a_i} , \ldots , 2\alpha^{(n-1)a_i} ).$$ 
 Clearly the vectors $v^{a_i}+v^{(a_i+\frac{n}{2})}$ and $v^{a_i}-v^{(a_i+\frac{n}{2})}$ are linearly independent over $K$. Moreover, 
$$v^{a_i}P_\gamma =(1 , \alpha^{3a_i}, \alpha^{2a_i}, \alpha^{5a_i}, \ldots, \alpha^{a_i} ) $$ 
and 
$$v^{(a_i+\frac{n}{2})}P_\gamma =(1 , \alpha^{3a_i+\frac{n}{2}}, \alpha^{2a_i}, \alpha^{5a_i+\frac{n}{2}}, \ldots \ \alpha^{a_i+\frac{n}{2}} ). $$ 
Since $P_\gamma$ does not change the entries in even number columns, $v^{a_i}+v^{(a_i+\frac{n}{2})}=v^{a_i}P_\gamma+v^{(a_i+\frac{n}{2})}P_\gamma $. Moreover, 
$$v^{a_i}P_\gamma - v^{(a_i+\frac{n}{2})}P_\gamma =(0, 2\alpha^{3a_i}, 0, 2\alpha^{5a_i}, 0, 2\alpha^{7a_i}, \ldots \ 2\alpha^{a_i} )=\alpha^{2a_i}(v^{a_i}-v_{(a_i+\frac{n}{2})}).$$ 
This completes the proof by showing that $\{v^{a_i},v^{(a_i+\frac{n}{2})}\}$ and $\{v^{a_i}P_\gamma ,v^{(a_i+\frac{n}{2})}P_\gamma \}$ generate the same vector space over $K$ for each $1\le i \le r$. 
\end{proof}

Next, we combine the results of the above lemmas and state a sufficient condition for permutation equivalence of linear cyclic codes. This is our main result of this section.

\begin{theorem}\label{new permutation}
Let $n$ be a positive integer divisible by $8$ and $q$ be a prime power such that $q\equiv 1 \pmod 4$ and $\gcd(n,q)=1$. Let $A\subseteq \Z/n\Z$ be a union of $q$-cyclotomic cosets modulo $n$ such that for each $a \in A$ either $a\in\{0,\frac{n}{2}\}$, or $(a+\frac{n}{2}) \bmod n$ is also an element of $A$. Then length $n$ linear cyclic codes with the defining sets $A\cup\{\frac{n}{4}\}$ and $A\cup \{\frac{3n}{4}\}$ over $\F_q$ are permutation equivalent under the action of $P_\gamma$.
\end{theorem}

\begin{proof}
Let $C_1$ and $C_2$ be length $n$ linear cyclic codes with the defining sets $A\cup\{\frac{n}{4}\}$ and $A\cup \{\frac{3n}{4}\}$ over $\F_q$, respectively. Suppose that the matrix $H$ is a generalized parity check matrix for $C_1$, in the form of $(\ref{PC matrix of linear cyclic codes})$. 
The proof follows from Lemmas $\ref{trivial action1}$ and $\ref{trivial action2}$ as the matrix $HP_\gamma$ is a generalized parity check matrix of $C_2$. Therefore, by Lemma $\ref{generalized equivalence}$, $C_1$ and $C_2$ are permutation equivalent under the action of $P_\gamma$.
\end{proof}

The next example shows that permutation equivalent cyclic codes of Theorem $\ref{new permutation}$ are not necessarily affine equivalent. Moreover, as we mentioned earlier, $P_\gamma$ is not of a generalized multiplier type. 

\begin{example}
Let $n=8$ and $A=\{0,1,5\}$. By Theorem $\ref{new permutation}$, length $8$ linear cyclic codes with the defining sets $A_1=A\cup \{2\}$ and $A_2=A\cup \{6\}$ over $\F_5$ are permutation equivalent. One can easily verify that there is no affine map between the sets $A_1$ and $A_2$.
Moreover, our computation in Magma \cite{magma} shows that all monomially equivalent length $8$ linear cyclic codes over $\F_5$ are either affine equivalent, permutation equivalent under the action of $P_\gamma$, or a combination of both. This classifies all the monomially equivalent cyclic codes of length $8$ over $\F_5$.
\end{example}

\subsubsection{New sufficient conditions for permutation equivalence of linear cyclic codes over $\F_4$}\label{S:n3}

Let $n$ be a positive odd integer divisible by $27$ and $\F_4=\{0,1,\omega,\omega^2\}$ be the field of four elements, where $\omega^2=\omega+1$. For the rest of this section, $\alpha$ denotes a fixed primitive $n$-th root of unity in 
some extension of $\F_4$ denoted $K$, such that $\alpha^{\frac{n}{3}}=\omega$. As before, $v^s\in K^n$ is defined by
$v^s=(1 , \alpha^s, \alpha^{2s}, \ldots , \alpha^{(n-1)s})$ for any $0 \le s \le n-1$. 
We define the permutation $\chi$ on $\Z/n\Z$ by
$$\chi(i)=\begin{cases} i+3 & \text{if}\ i \equiv 0 \ \text{or}\ 4\ \text{or}\ 5 \pmod 9 \\ i-3 & \text{if}\ i \equiv 3 \ \text{or}\ 7 \ \text{or}\ 8 \pmod 9 \\
i & \text{otherwise.}
 \end{cases}$$ 
Note that since $\chi(0) \neq 0$ the permutation $\chi$ is not of  multiplier 
or  generalized multiplier type.
We denote the permutation matrix corresponding to $\chi$ by $P_\chi$. 
Let $\{s_i: 0 \le i \le n-1\}$ be the standard basis of $\F_4^n$. Then 
$$s_iP_\chi= 
\begin{cases} s_{i+3} & \text{if}\ i \equiv 0 \ \text{or}\ 4\ \text{or}\ 5 \pmod 9 \\ s_{i-3} & \text{if}\ i \equiv 3 \ \text{or}\ 7 \ \text{or}\ 8 \pmod 9 \\
i & \text{otherwise.}
 \end{cases}$$
Since $3 \mid n$ there are always three singleton $4$-cyclotomic cosets modulo $n$ namely $\{0\}$, $\{\frac{n}{3}\}$, and $\{\frac{2n}{3}\}$. Before stating our main result of this section, we have three intermediate lemmas.

\begin{lemma}\label{F4 permutation 1}
Let $n$ be a positive odd integer divisible by $27$ and $C$ be a length $n$ linear cyclic code over $\F_4$ with the defining set $\{a\}$, where $a\in \{0, \frac{n}{3},\frac{2n}{3}\}$. Then $C=CP_\chi$.
\end{lemma}

\begin{proof}
Note that $v^a=(1, \alpha^{a}, \alpha^{2a}, \ldots, \alpha^{(n-1)a} )$ is a parity check matrix for $C$. Moreover, the entry in $j$-th column of $v^a$ is
$$(v^a)_j=\begin{cases} 1 & \text{if}\ a=0 \\ \omega^j & \text{if}\ a=\frac{n}{3}\\
\omega^{2j}& \text{if} \ a=\frac{2n}{3}.
 \end{cases}$$ 
One can easily see that the $j$-th column of $v^a$ remains unchanged after the action of $P_\chi$. Thus $v^a=v^aP_\chi$. Therefore, by Lemma $\ref{generalized equivalence}$, we have $C=CP_\chi$.
\end{proof}
Recall that $Z(a)$ denotes the $4$-cyclotomic coset modulo $n$ containing $a$. 

\begin{lemma}\label{F4 permutation 2}
Let $n$ be a positive odd integer divisible by $27$, and $C_1$ and $C_2$ be length $n$ linear cyclic codes over $\F_4$ with defining sets $Z(\frac{n}{9})$ and $Z(\frac{2n}{9})$, respectively.
Then $C_2=C_1P_\chi$.
\end{lemma}

\begin{proof}
First, note that $Z(\frac{n}{9})=\{\frac{n}{9},\frac{4n}{9},\frac{7n}{9}\}$ and $Z(\frac{2n}{9})=\{\frac{2n}{9},\frac{5n}{9},\frac{8n}{9}\}$.
Let $u=v^{\frac{n}{9}}+v^{\frac{4n}{9}}+v^{\frac{7n}{9}}$. Then for each $0\le j \le n-1$, we have
$$u_j=\alpha^{j\frac{n}{9}}+\alpha^{j\frac{4n}{9}}+\alpha^{j\frac{7n}{9}}
=\alpha^{j\frac{n}{9}}(1+\alpha^{j\frac{n}{3}}+ \alpha^{j\frac{2n}{3}}).$$
Thus
$$u_j=\begin{cases} \alpha^{j\frac{n}{9}} & \text{if}\ 3 \mid j \\ 0 & \text{otherwise}
 \end{cases}=\begin{cases} 1& \text{if}\ j \equiv 0 \pmod 9\\ \omega & \text{if}\ j \equiv 3 \pmod 9\\ 
 \omega^2 & \text{if}\ j \equiv 6 \pmod 9\\
 0 & \text{otherwise}
 \end{cases}$$ 
 for each $0\le j \le n-1$.
A similar computation shows that if $x=v^{\frac{2n}{9}}+v^{\frac{5n}{9}}+v^{\frac{8n}{9}}$, then  
 $$x_j=\begin{cases} 1& \text{if}\ j \equiv 0 \pmod 9\\ \omega^2 & \text{if}\ j \equiv 3 \pmod 9\\ 
 \omega & \text{if}\ j \equiv 6 \pmod 9\\
 0 & \text{otherwise}
 \end{cases}$$
 for each $0\le j \le n-1$.
Let $u'$ and $u''$ be the cyclic shifts of $u$ 
by one and two positions to the right, respectively.
The set $S=\{u,u',u''\}$ is linearly independent over $K$ and matrix $H$ consisting of elements of $S$ as its rows is a parity check matrix for $C_1$. 
 Next, we show that if $x'$ and $x''$ 
are the cyclic shifts of $x$ 
by one and two positions to the right, respectively,
then $\{uP_\chi, u'P_\chi, u'' P_\chi\}=\{\omega x,\omega x',\omega x''\}$. This implies that $HP_\gamma$ is a parity check matrix for $C_2$. 
A straightforward computation shows that
$$(uP_\chi)_j=\begin{cases} \omega & \text{if}\ j \equiv 0 \pmod 9\\ 1 & \text{if}\ j \equiv 3 \pmod 9\\ 
 \omega^2 & \text{if}\ j \equiv 6 \pmod 9\\
 0 & \text{otherwise}
 \end{cases}$$
 for each $0\le j \le n-1$. Thus $uP_\chi=\omega x$. The equalities $u'P_\chi=\omega x'$ and $u''P_\chi=\omega x''$ follow accordingly. Hence, $HP_\gamma$ is a parity check matrix for $C_2$ and by Lemma $\ref{generalized equivalence}$, 
and we have $C_2=C_1P_\chi$. 
\end{proof}

\begin{lemma}\label{F4 permutation 3}
Let $n=3^tk$, where $k$ and $t\ge 3$ are positive integers such that $\gcd(k,3)=1$. Let $1\le e \le n-1$.
Then the following statements hold.
\begin{enumerate}
\item The set $A_e=\{(ek+i\frac{n}{3^{t-1}}) \bmod n: 0\le i \le 3^{t-1}-1 \}$ is a union of $4$-cyclotomic cosets modulo $n$.
\item  Suppose that $C$ is the length $n$ linear cyclic code over $\F_4$ with defining set $A_e$. Then $C=CP_\chi$.
\end{enumerate}
\end{lemma}

\begin{proof}
To prove (1) we show that $A_e=\displaystyle\bigcup_{i=0}^{3^{t-1}-1}Z(ek+i\frac{n}{3^{t-1}})$, where 
some of the cyclotomic cosets are repeating in the union.
 Obviously, $A_e \subseteq \displaystyle\bigcup_{i=0}^{3^{t-1}-1}Z(ek+i\frac{n}{3^{t-1}})$. Let $0\le s \le 3^{t-1}-1$ and $j\geq 0$ be an arbitrary integer.  We have  $4^j=3l+1$ for some positive integer $l$ and
$$4^j(ek+s\frac{n}{3^{t-1}})=ek+s\frac{n}{3^{t-1}}+ 3lek+3ls\frac{n}{3^{t-1}}=ek+(s+le+3ls)\frac{n}{3^{t-1}} \in A,$$
where the last equality follows form the fact that $3k=\frac{n}{3^{t-1}}$. Hence $Z(ek+s\frac{n}{3^{t-1}}) \subseteq A$ for any $0\le s \le 3^{t-1}-1$. This implies that $A=\displaystyle\bigcup_{i=0}^{3^{t-1}-1}Z(ek+i\frac{n}{3^{t-1}})$.

Now we prove (2). Let $H$ be a generalized parity check matrix for $C$ consisting of the row vectors
$$v^{(ek+i\frac{n}{3^{t-1}})}=(1, \alpha^{ek+i\frac{n}{3^{t-1}}}, \alpha^{2(ek+i\frac{n}{3^{t-1}})}, \ldots, \alpha^{(n-1)(ek+i\frac{n}{3^{t-1}})} )$$
 for $0\le i \le 3^{t-1}-1$. Next, we produce a new generalized parity check matrix $H'$ for $C$ using linear combinations of rows of $H$. We also show that the rows of $H'P_\chi$ are just permutations of the rows of $H'$. This implies that $C=CP_\chi$. 

Let $u=\displaystyle\sum_{i=0}^{3^{{t-1}}-1} v^{(ek+i\frac{n}{3^{t-1}})}$ be the sum of 
all  rows of $H$. For each $0\le l \le n-1$, we can compute the entry of the $l$-th column of $u$ as
\begin{equation}\label{E:new form}
u_l=\sum_{i=0}^{3^{t-1}-1}\alpha^{l(ek+i\frac{n}{3^{t-1}})}=\alpha^{ekl}(\displaystyle\sum_{i=0}^{3^{{t-1}}-1} \alpha^{\frac{iln}{3^{t-1}}})=\begin{cases} \alpha^{ekl} & \text{if}\ 3^{t-1} \mid l \\ 0 & \text{otherwise.}
 \end{cases}
 \end{equation}
 Let $u^{[m]}$ be $m$ positions cyclic shift of the vector $u$ to the right for each $0\le m \le 3^{t-1}-1$. Then $(\ref{E:new form})$ implies that the set $B=\{u^{[m]}: 0\le m\le 3^{t-1}-1\}$ is linearly independent over $K$. Let $H'$ be a $3^{t-1} \times n$ matrix over $K$ consisting of elements of $B$ as its rows. Then $H'$ is a generalized parity check matrix for $C$. Moreover, for any $0\le m \le 3^{t-1}-1$ and $0\le l \le n-1$ we have
 $$(u^{[m]})_l=\begin{cases} \alpha^{ekl} & \text{if}\ 3^{t-1} \mid l-m \\ 0 & \text{otherwise.}
 \end{cases}$$
 Since $9 \mid 3^{t-1}$, 
 $$ u^{[m]} P_\chi=\begin{cases} u^{[m+3]}& \text{if}\ m \equiv 0 \ \text{or}\ 4\ \text{or}\ 5 \pmod 9 
 \\ u^{[m-3]} & \text{if}\ m \equiv 3 \ \text{or}\ 7 \ \text{or}\ 8 \pmod 9 \\
 u^{[m]} & \text{otherwise.}
 \end{cases}$$
 Hence the rows of $H'P_\chi$ are just a permutation of the rows of $H'$. Therefore $C=CP_\chi$.
 \end{proof}
 
 Now we combine all the above results and state a sufficient condition for the permutation equivalence of linear cyclic codes over $\F_4$. Let $n=3^tk$, where $k$ and $t\ge 3$ are positive integers such that $\gcd(k,3)=1$.
We define $A_e$
as in part (1) of Lemma~\ref{F4 permutation 3}
for any $1\le e \le n-1$. The next theorem is our main result of this section.
 

  \begin{theorem}\label{F4 permutation}
Let $n=3^tk$, where $k$ and $t\ge 3$ are positive integers such that $\gcd(k,3)=1$.
Suppose that $B \subseteq \{0, \frac{n}{3},\frac{2n}{3}\}$ and $1\le e_j \le n-1$ for $1\le j \le r$, where $r$ is a positive integer. Then length $n$ linear cyclic codes with the defining sets $T_1=Z(\frac{n}{9}) \cup B \cup \displaystyle\bigcup_{j=1}^{r} A_{e_j}$ and $T_2=Z(\frac{2n}{9}) \cup B \cup \displaystyle\bigcup_{j=1}^{r} A_{e_j}$ are permutation equivalent over $\F_4$ under the action of $P_\chi$.
 \end{theorem}

\begin{proof}
Let $C_1$ and $C_2$ be linear cyclic codes of length $n$ over $\F_4$ with the defining sets $T_1$ and $T_2$, respectively.
By Lemmas $\ref{F4 permutation 1}$, $\ref{F4 permutation 2}$, and $\ref{F4 permutation 3}$ the matrix $P_\chi$ maps a generalized parity check matrix of $C_1$ to a generalized parity check matrix of $C_2$. Therefore, the result follows from Lemma $\ref{generalized equivalence}$.
\end{proof}

Next, we present an application of the above result.

\begin{example}
Let $n=27$. An easy computation shows that $Z(1)=\{1+i\frac{n}{9}: 0\le i \le 8 \}$ and $Z(2)=\{2+i\frac{n}{9}: 0\le i \le 8 \}$. Now by Theorem $\ref{F4 permutation}$, the linear cyclic codes of length $27$ over $\F_4$ with the following 
pairs of defining sets are permutation equivalent: 
\begin{itemize}
\item $Z(0) \cup Z(1) \cup Z(3)$ and $Z(0) \cup Z(1) \cup Z(6)$.
\item $Z(0) \cup Z(2) \cup Z(3)$ and $Z(0) \cup Z(2) \cup Z(6)$. 
\item $Z(0) \cup Z(1) \cup Z(3) \cup Z(9)$ and $Z(0) \cup Z(1) \cup Z(6) \cup Z(9)$. 
\item $Z(0) \cup Z(1) \cup Z(3) \cup Z(9) \cup Z(18)$ and $Z(0) \cup Z(1) \cup Z(6) \cup Z(9) \cup Z(18)$.
\end{itemize}
The above pairs 
of codes are not permutation equivalent under the action of multipliers, generalized multipliers, or a combination of both. This is mainly because multipliers and generalized multipliers always map $0$ to $0$. However, $\chi(0)=3$. Moreover, the above pairs of codes are not affine equivalent. There are many more such pairs of permutation equivalent linear cyclic codes  over $\F_4$ of length $27$. 
To the best of our knowledge, the permutation equivalence
of the above pairs of codes can not be proved by earlier results
in the literature.
\end{example}


\begin{example} This example presents some other values of $n$ and $q$ such that there exist at least a pair of monomially equivalent linear cyclic codes over $\F_q$ of length $n$ which are not affine equivalent. 
\begin{itemize}
\item For $q=2$, lengths $n=45,49$.
\item For $q=3$, lengths $n=8^\ast,16^{\ast\ast},32^{\ast\ast},40^\ast, 48^{\ast\ast}, 56^\ast$.
\item For $q=4$, lengths $n=25,27^\diamond,49$.
\item For $q=5$, lengths $n=\underline{8},\underline{16},\underline{24}^\ast$.
\item For $q=7$, lengths $n=8^\ast,16^\ast,18,24^\ast,32^{\ast\ast}, 40^\ast$.
\item For $q=11$, lengths $n=8^\ast,16^{\ast\ast},24^\ast$.
\end{itemize} 
In the above list:
\begin{itemize}
\item $\ast$ shows the code lengths for which there exist a pair of monomially equivalent codes obtained by Theorem $\ref{monomial action}$.

\item The underline shows the code lengths for which there exist a pair of permutation equivalent codes obtained by Theorem $\ref{new permutation}$.

\item $\diamond$ shows the code lengths for which there exist a pair of permutation equivalent codes obtained by Theorem $\ref{F4 permutation}$.
\item $\ast\ast$ shows the code lengths for which there exist a pair of monomially equivalent codes obtained by the action of the monomial matrix 
\begin{equation*}
\begin{bmatrix}
 P_\sigma D& \cdots&0 \\
 \vdots& \ddots &\vdots \\
 0&\cdots & P_\sigma D
\end{bmatrix}
\end{equation*}
consisting of more than one block of the matrix $P_\sigma D$ of Theorem $\ref{monomial action}$ on the main diagonal. 
\end{itemize}
\end{example}

The sufficient conditions of Theorems $\ref{monomial action}$, $\ref{new permutation}$, and $\ref{F4 permutation}$ along with the monomial and permutation equivalence criteria given in Section \ref{cyclic intro} detect many monomially equivalent linear cyclic codes over various finite fields. Therefore, it is extremely beneficial and cheap to use such conditions, and make the search for new linear cyclic codes with good parameters more systematic.

\subsection{More results on the equivalence of linear cyclic codes}\label{conjecture section}

Conjecture $2$ of \cite{aydin2019equivalence} 
proposes that the condition ``$n$ divides $|A_1|(q-1)b$'' of Theorem $\ref{cyclic2}$ is a necessary condition for the isometric equivalence of linear cyclic codes by a shift map. The next theorem proves this fact and also strengthens the result of Theorem $\ref{cyclic2}$. Recall that $\alpha$ is a fixed primitive $n$-th root of unity in the field $K$ which is a finite field extension of $\F_q$.

\begin{theorem}\label{affine conjecture}
Let $C_1$ and $C_2$ be two linear cyclic codes over $\F_q$ of length $n$ with defining sets $A_1$ and $A_2$, respectively and $b$ be a positive integer. The codes $C_1$ and $C_2$ are isometrically equivalent through the shift map $\phi_b$  on their defining sets if and only if $\phi_b$ is
a bijection between $A_1$ and $A_2$
and $n$ divides $|A_1|(q-1)b$.
\end{theorem}

\begin{proof}
Without loss of generality assume $\phi_b(A_1)=A_2$.
We only prove the forward direction as the reverse direction follows from Theorem $\ref{cyclic2}$.
Let $g_1(x)$ and $g_2(x)$ be the generator polynomials of $C_1$ and $C_2$, respectively. 
Then $g_1(x)=\displaystyle\prod_{i\in A_1}(x-\alpha^i)$ and 
\begin{equation}\label{shift generator}
g_2(x)=\prod_{i\in A_2}(x-\alpha^i)=\prod_{i\in A_1}(x-\alpha^{i+b})=\alpha^{b|A_1|}\prod_{i\in A_1}(\alpha^{-b}x-\alpha^i)=\alpha^{b|A_1|} g_1(\alpha^{-b}x).
\end{equation}
Since $g_2(x)$ is defined over $\F_q$, the equality $\big(g_2(0)\big)^q=g_2(0)$ holds. By combining this fact with $(\ref{shift generator})$, we get 
\begin{equation}\label{constantterm}
0=\big(g_2(0)\big)^q-g_2(0)=\alpha^{b|A_1|q} \big(g_1(0)\big)^q-\alpha^{b|A_1|} g_1(0).
\end{equation}
The generator polynomial of a non-empty linear cyclic code always has a non-zero constant term. Let the non-zero constant term of $g_1(x)$ be $g_1(0)=c \in \F_q$. Now $(\ref{constantterm})$ implies that 
\begin{equation}\label{shift equ 3}
\alpha^{b|A_1|q}c^q -\alpha^{b|A_1|} c=c\alpha^{b|A_1|}(\alpha^{b|A_1|(q-1)}-1)=0.
\end{equation}
Since $c$ and $\alpha^{b|A_1|}$ are both non-zero, $(\ref{shift equ 3})$ implies that $\alpha^{b|A_1|(q-1)}=1$ or equivalently $n \mid b|A_1|(q-1)$.
\end{proof}

Let $g_1(x)$ and $g_2(x)$ be generator polynomials of two monomially equivalent cyclic codes of length $n$ over $\F_q$. Conjecture $1$ of \cite{aydin2019equivalence} proposes that for any integer $m\geq 1$ coprime to $q$, the length $nm$ cyclic codes generated by $g_1(x)$ and $g_2(x)$ are monomially equivalent. Indeed this statement holds for linear cyclic codes. However note that the linear cyclic code of length $nm$ generated by $g(x)$ has minimum distance of at most two since $x^n-1$ is in correspondence with a weight two codeword. As such low minimum distance cyclic codes are not very interesting for us, we only state this result without proof. 

\begin{proposition}\label{monomial conjecture}
Let $g_1(x)$ and $g_2(x)$ be generator polynomials of two monomially (respectively permutation) equivalent cyclic codes over $\F_q$ of length $n$. For any integer $m\geq 1$ coprime to $q$, the length $nm$ cyclic codes generated by $g_1(x)$ and $g_2(x)$ are also monomially (respectively permutation) equivalent.
\end{proposition}

\section{Equivalence of linear constacyclic codes over $\F_4$}\label{constacyclic code equivalence}

Constacyclic codes are one of the first well-known generalizations of linear cyclic codes with many similar properties. More information about constacyclic codes can, for example, be found in \cite{Kai} and \cite{Krishna}. To keep the statements of this section simple and also apply the results of this section to prune the search algorithm for new binary quantum codes, we only restrict our attention to constacyclic codes over $\F_4$.

 

\subsection{A short overview on equivalence of constacyclic codes over $\F_4$}\label{constacyclic1}

Let $\F_4=\{0,1,\omega,\omega^2\}$ be the field of four elements, where $\omega^2=\omega+1$. Throughout this section, we assume that $n$ is a positive odd integer and $\delta$ is a fixed primitive $3n$-th root of unity in $L$ which is a finite field extension of $\F_4$ such that $\delta^n=\omega$.

For any $0\neq \eta \in \mathbb{F}_4$, a linear code $C \subseteq \mathbb{F}_4^n$ is called an \textit{$\eta$-constacyclic} code, if for any codeword $(a_0,a_1,\cdots,a_{n-1}) \in C$, the vector $(\eta a_{n-1},a_0,\cdots,a_{n-2})$ is also in $C$. 
If $\eta=1$, then $\eta$-constacyclic code is a cyclic code. 
 Similar to linear cyclic codes, there is a one-to-one correspondence between $\eta$-constacyclic codes of length $n$ over $\F_4$ and ideals of the ring $\mathbb{F}_4[x]/\langle x^n-\eta \rangle$. Therefore, each $\eta$-constacyclic code over $\F_4$ of length $n$ can be uniquely represented by a unique monic polynomial 
$g(x)$ over $\F_4$ such that $g(x)\mid x^n-\eta$. The polynomial $g(x)$ is called the {\em generator polynomial} of such $\eta$-constacyclic code. 
As in the case of linear cyclic codes,
the codewords of an $\eta$-constacyclic code over $\F_4$ are
$(a_0,\ldots,a_{n-1})$ where $a_0+\ldots+a_{n-1}x^{n-1}$
is a multiple of $g(x)$ by a polynomial over $\F_4$.

For any $a\in\F_4$, the {\em conjugate} of $a$ is defined by $\overline a=a^2$. 
By \cite[Theorem 3.2]{Chen},
the conjugation map 
\begin{equation}\label{two equivalent constacyclic}
\Theta:\mathbb{F}_4[x]/\langle x^n -\omega\rangle \rightarrow \mathbb{F}_4[x]/\langle x^n -\omega^2 \rangle 
 \end{equation}
defined by $\Theta(\sum_{i=0}^{n-1}a_ix^i)=\sum_{i=0}^{n-1}\overline{a_i}x^i$ is an isometry,  and it gives a one-to-one correspondence between $\omega$- and $\omega^2$-constacyclic codes over $\F_4$ of length $n$. 
Therefore, from now on, we restrict our attention only to $\omega$-constacyclic codes over $\mathbb{F}_4$. 
When $\gcd(n,3)=1$, another isometry of linear codes which also gives a one-to-one correspondence between $\omega$-constacyclic codes and linear cyclic codes of length $n$ over $\F_4$ is provided in \cite[Theorem 15]{bierbrauer}. 
Thus only when $3\mid n$, $\omega$-constacyclic codes and cyclic codes of length $n$ over $\F_4$ can have different parameters.

%

%

All the roots of $x^n-\omega$ belong to the field $L$ and are in the form of $\{\delta^{1+3j}:0\le j \le n-1\}$.
Hence an $\omega$-constacyclic code with the generator polynomial $g(x)$ can be represented by its unique {\em defining set} defined by
$$\{1+3j: 0\le j \le n-1 \ \text{and}\ g(\delta^{1+3j})=0 \}.$$
The following lemma states another useful isometry of $\omega$-constacyclic codes.
\begin{lemma}\cite[Lemma $2.4$]{aydin2017some}\label{consta isometry}\label{affine1}
Let $n$ and $e\equiv 1 \pmod 3$ be positive integers such that $\gcd(3n,e)=\gcd(n,2)=1$. Then the $\F_4$-isomorphism ${\psi}:\mathbb{F}_4[x]/\langle x^n -\omega\rangle \rightarrow \mathbb{F}_4[x]/\langle x^n -\omega \rangle$ defined by $\psi(f(x))=f(x^e) \bmod{(x^n-\omega)} $ maps each $\omega$-constacyclic code to an isometric (permutation) equivalent $\omega$-constacyclic code. 
\end{lemma}

Let $A_1$ and $A_2$ be defining sets of two length $n$ $\omega$-constacyclic codes over $\F_4$. 
By Lemma $\ref{consta isometry}$, if there exists a multiplier $\mu_{e}$ defined on $\mathbb Z/3n\mathbb Z$ such that $\mu_e(A_1)=A_2$, then the $\omega$-constacyclic codes with the defining sets $A_1$ and $A_2$ are isometrically equivalent.  

\begin{theorem} \cite[Theorem 3.8.8]{Huffman}\label{Going down}
Let $C$ be an $[n,k,d]$ linear code over $\F_{4^t}$ for some $t \geq 1$. If $C$ has a basis of codewords over $\F_4$, then $C\vert_{\F_4}$, the restriction of $C$ to the codewords over $\F_4$, is an $[n,k,d]$ linear code over $\F_4$.  
\end{theorem} 

The code $C\vert_{\F_4}$ in the above theorem is known as a {\em subfield subcode} of the code $C$. 

\subsection{New results on equivalence of constacyclic codes over $\F_4$}\label{constacyclic2}

Our initial goal of this section is to show that certain affine bijections on the defining set of $\omega$-constacyclic codes over $\F_4$ produce codes with the same parameters. 
We denote the defining set of an $\omega$-constacyclic code with the generator polynomial $g(x)$ by $A_g$.

\begin{theorem}\label{affine2}
Let $C_1$ and $C_2$ be two $\omega$-constacyclic codes over $\mathbb{F}_4$ of length $n$ with the generator polynomials $g(x)$ and $h(x)$, respectively. Let $1\le j \le n$ be a positive integer such that $n$ divides $3j\deg(g(x))$. If the shift map $\phi_{3j}(x)$ defined on $\Z/3n\Z$ satisfies $\phi_{3j}(A_g)=A_h$, then $C_1$ and $C_2$ have the same parameters. 
\end{theorem}

\begin{proof}
Consider the map $\Theta: L[x]/\langle x^n-\omega \rangle \rightarrow L[x]/\langle x^n-\omega \rangle$ defined by $\Theta (f(x))=f(\delta^{-3j}x)$. First note that $\Theta(x^n-\omega)=(\delta^{-3j})^nx^n-\omega=x^n-\omega$. This shows that $\Theta$ is well-defined. Moreover, since $\Theta$ is a substitution map, it is  also 
a homomorphism and preserves the Hamming weight of vectors. 
The ring $L[x]/\langle x^n-\omega \rangle$ has finitely many elements, and to prove $\Theta$ is a bijection, it suffices to show that $\Theta$ 
is surjective. Let $f(x) \in L[x]/\langle x^n-\omega \rangle$. Then $f(\delta^{3j}x)\in L[x]/\langle x^n-\omega \rangle$ and $\Theta(f(\delta^{3j}x))=f(x)$. This implies that $\Theta$ is an isomorphism. 

We can factorize $g(x)$ as $g(x)=\displaystyle\prod_{i\in A_{g}} (x-\delta^i)$ over $L$. Then 
\begin{equation}\label{isometry factorization}
\Theta(g(x))=\prod_{i\in A_{g}}(\delta^{-3j}x-\delta^i)=(\delta^{-3j})^{|A_{g}|}\prod_{i\in A_{g}}(x-\delta^{(i+3j)})=a\prod_{i\in A_{g}}(x-\delta^{(i+3j)})=ah(x)
\end{equation}
where $a=(\delta^{-3j})^{|A_{g}|} \in \F_4$ as $n$ divides $3j\deg(g(x))$. 
Note that in the last equality of (\ref{isometry factorization})
we used the assumption that $\phi_{3j}(A_g)=A_h$.
 Thus $\Theta$ gives an isometry of linear codes between the $\omega$-constacyclic codes over $L$ of length $n$ generated by $g(x)$ and $h(x)$. 
Both of the mentioned codes have a basis over $\F_4$. Thus by Theorem $\ref{Going down}$, their subfield subcodes over $\F_4$, which are $C_1$ and $C_2$, have the same parameters.
\end{proof}
 

We give a generalization of  Theorem~\ref{affine2}.

\begin{corollary}\label{affine consta}
Let $C_1$ and $C_2$ be two $\omega$-constacyclic codes over $\mathbb{F}_4$ of length $n$ with the generator polynomials $g(x)$ and $h(x)$, respectively. If there exists a map $\theta(x)=(ex+3j) \bmod {(3n)}$ on $\Z/3n\Z$ such that 
\begin{itemize}
\item $e=3k+1$ and $\gcd(3n,e)=1$, 
\item $n$ divides $3j\deg(g(x))$,
\item and $\theta(A_g)=A_h$,
\end{itemize}
 then $C_1$ and $C_2$ have the same parameters. 
\end{corollary}

\begin{proof}
The proof follows from Lemma $\ref{affine1}$ and Theorem $\ref{affine2}$.
\end{proof}

Next, we show that the condition ``$n$ divides $ 3j\deg(g(x))$'' of Theorem $\ref{affine2}$ is a necessary condition for the existence of a shift bijection between the defining sets of two $\omega$-constacyclic codes over $\F_4$. 
\begin{theorem}\label{shift sufficient}
Let $A_1$ and $A_2$ be defining sets of two $\omega$-constacyclic codes of length $n$ over $\F_4$ and $b$ be a positive integer. If the shift map $\phi_b(x)$ defined on $\Z/3n\Z$ satisfies $\phi_b(A_1)=A_2$, then $3$ divides $b$ 
and $n$ divides $ b |A_1|$.
\end{theorem}

\begin{proof}
As we mentioned earlier, roots of $x^n-\omega$ are $\delta^{3k+1}$ for $0\le k\le n-1$. Thus, 
$3$ divides $b$ 
 as otherwise for each $s \in \phi_b(A_1)$, $\delta^s$ is not a root of $x^n-\omega$.  
Moreover, $x^n-\omega \mid x^{3n}-1$ and therefore defining set of each $\omega$-constacyclic code over $\F_4$ of length $n$ is also defining set of a length $3n$ linear cyclic code over $\F_4$. Now Theorem $\ref{affine conjecture}$ implies that $3n \mid 3b|A_1|$ which is equivalent to $n\mid b |A_1|$.
\end{proof}

Our final goal of this section is to show that all permutation equivalent $\omega$-constacyclic codes of length $n$ over $\F_4$ such that $(3n,\phi(3n))=1$ are given by the action of multipliers on their defining sets.
This result is analogous to the result of Theorem $\ref{isoequivalent}$ for linear cyclic codes.

\begin{theorem}\label{Palfy consta}
Let $C_1$ and $C_2$ be two non-trivial $\omega$-constacyclic codes over $\F_4$ of length $n$ with defining sets $A_1$ and $A_2$ such that $\gcd(3n,\phi(3n))=1$. Then $C_1$ and $C_2$ are permutation equivalent if and only if there exists a multiplier $\mu_e$ defined on $\mathbb Z/3n\mathbb Z$ such that $\mu_e(A_1)=A_2$ for some positive integer $e\equiv 1 \pmod 3$.
\end{theorem}

\begin{proof}
 
We only prove the forward direction as the reverse follows from Lemma $\ref{affine1}$.
Since $x^n-\omega \mid x^{3n}-1$,
the sets $A_1$ and $A_2$ are also defining sets of two linear cyclic codes $D_1$ and $D_2$ of length $3n$ over $\F_4$, respectively. 
It is enough to show that $D_1$ and $D_2$ are permutation equivalent, and since $(3n,\phi(3n))=1$, Theorem \ref{isoequivalent} implies the existence of a multiplier $\mu_e$ defined on $\mathbb Z/3n\mathbb Z$ such that $\mu_e(A_1)=A_2$. Moreover, the fact that $A_1,A_2 \subset \{3k+1: 0\le k \le n-1\}$ implies $e\equiv 1 \pmod 3$. 

Let $A_1=\{a_1,a_2,\cdots,a_r\}$ and the matrix  
$$H_1=\begin{bmatrix}
1& \delta^{a_1}& \cdots & \delta^{a_1(n-1)} \\
\vdots&\vdots & \cdots & \vdots \\
1& \delta^{a_r}& \cdots & \delta^{a_r(n-1)} \\
\end{bmatrix} $$
be a generalized parity check matrix for the code $C_1$. Since $\delta^n=\omega$
the matrix
\begin{equation}\label{GPC1}
H=\begin{bmatrix}
1& \delta^{a_1}& \cdots & \delta^{a_1(3n-1)} \\
\vdots&\vdots & \cdots & \vdots \\
1& \delta^{a_r}& \cdots & \delta^{a_r(3n-1)} \\
\end{bmatrix}=\begin{bmatrix}
H_1&\omega H_1&\omega^2H_1\\
\end{bmatrix}
\end{equation}
is a generalized parity check matrix for $D_1$. 
As the codes $C_1$ and $C_2$ are permutation equivalent, by Lemma $\ref{generalized equivalence}$, there exists a permutation matrix $P$ such that $H_1P$ is a generalized parity check matrix for $C_2$. Let 
\begin{equation}\label{matrix m}
P_{3}=\begin{bmatrix}
P& 0&0 \\
0& P&0 \\
0& 0&P \\
\end{bmatrix}
\end{equation}
be the $3n\times 3n$ permutation matrix over $\F_4$ containing $3$ copies of $P$ on the main diagonal. Let $g_2(x)$ be the generator polynomial of $C_2$. By Lemma $\ref{generalized equivalence}$, $D_1$ and the linear code over $\F_4$ with the generalized parity check matrix $HP_{3}$ are permutation equivalent. Next, we show that $HP_{3}$ is also a  generalized parity check matrix for $D_2$.   

We denote the length $n$ and $3n$ column vectors corresponding to a polynomial $f(x)\in \F_4[x]/\langle x^n-1\rangle$ by $\left[ f(x)\right]_n$ and $\left[ f(x)\right]_{3n}$, respectively.
Note that 
\begin{equation*} HP_{3}=\begin{bmatrix}
 H_1P & \omega H_1P &\omega^2 H_1P\\ \end{bmatrix}.
 \end{equation*}
 Let $g_2(x)=\sum_{j=0}^{n-1}b_jx^j$ be the generator polynomial of the code $D_2$ and $0\le i \le 3n-1$. The vector $\left[ x^i g_2(x) \right]_{3n}$ has at most $n$ non-zero coordinates. Let $i=rn+s$, where $0\le s<n$ and $0 \le r \le 2$.
 Then \begin{equation}\label{equ3}
\begin{split}
 (HP_{3})\left[ x^i g_2(x) \right]_{3n}&=
 \omega^r H_1P \left[x^sg_2(x) \right]_n=0.
 \end{split}
\end{equation}
The last equality of $(\ref{equ3})$ follows from the fact that $H_1P$ is a generalized parity check matrix for $C_2$. 
This shows that  $HP_{3}$ is a  generalized parity check matrix for $D_2$. Hence
 $D_1$ and $D_2$ are permutation equivalent and this completes the proof. 
\end{proof}

%

\section{New linear and quantum codes}\label{applications}
The results of this paper make the search algorithm for new linear 
classical and quantum codes more efficient. In practice, the parameters of quantum codes are known much less than classical linear codes in the literature,
as can be seen in the tables \cite{Grassl}. 
Results about equivalence of linear codes
can be used to prune branches of the search algorithms
for new good codes.
Hence pruning the search algorithm for linear codes also helps to discover more record-breaking quantum codes.
In particular, we provide examples of two new quantum codes and one new linear code over $\F_4$ with a better minimum distance than the previous best codes with the same length and dimension. 

The following connection between quaternary linear codes and binary quantum codes will be used in the construction of our new quantum codes. The Hermitian dual of a linear code $C$ over $\F_4$ is denoted by $C^{\bot_h}$.

\begin{theorem}\cite[Theorem 2]{Calderbank}\label{quantum def}
 Let $C$ be a linear $[n,k,d]$ code over $\F_4$ such that $C^{\bot_h}\subseteq C$. Then we can construct an $[[n, 2k-n,d']]$ {binary quantum code}, where $d'$
 is the minimum weight in $C\setminus C^{\bot_h}$.
 If $C= C^{\bot_h}$, then $d'=d$.
 \end{theorem}

Our main tool for construction
of new quantum codes is the {\em nearly self-orthogonal construction} of quantum codes provided below. 
\begin{theorem}\cite[Theorem 2]{Lisonek}\label{lisonek construction}
Let $C$ be an ${[n,k]}$ linear code over $\F_4$ and $e=n-k- \dim(C \cap C^{\bot_h})$. Then there exists a binary
quantum code with parameters $[[n+e,2k-n+e,d]]$, where
$$ d\geq \min\{d(C), d(C+ C^{\bot_h}) +1\}.$$ 
\end{theorem}


\begin{example}
Let $n=51$ and $A$ be the defining set of a linear cyclic code $C$ over $\F_4$ with the coset leaders $\{0,2,7,17,34\}$. The code $C$ is {affine equivalent to 24} linear cyclic codes of the same length over $\F_4$,
of which only one needs to be considered, thus reducing the running
time by a factor of 24.
Moreover, $C$ is a $[51,40]$ linear cyclic code and $\min\{d(C), d(C+ C^{\bot_h}) +1\}=4$. After applying the construction of Theorem $\ref{lisonek construction}$, we get $e=3$ which implies the existence of a $[[54,32]]$ binary quantum code $D$. 
Note that the dual containing code $D$ was constructed using the proof of Theorem \ref{lisonek construction} provided in \cite{Lisonek}.
 The minimum weight in $D\setminus D^{\bot_h}$ is $6$. Hence we get a $[[54,32,6]]$ which has a {\em better minimum distance} than the current best known binary quantum code with the same length and dimension. 
\end{example}

\begin{example}
Let $n=111$ and $A$ be the defining set of an $\omega$-constacyclic code $C$ over $\F_4$ with the coset leaders $\{19,37\}$. The code $C$ has the same parameters as {five} other $\omega$-constacyclic codes of the same length over $\F_4$ through an affine bijection on their defining sets. The code $C$ is a $[111,90]$ linear code, $e=3$, and $\min\{d(C), d(C+ C^{\bot_h}) +1\}=9$.
After applying the construction of Theorem $\ref{lisonek construction}$ to the 
code $C$, we constructed a new binary quantum code
with parameters $[[114,72,9]]$ which has a {\em better minimum distance} than the current best known binary quantum code. Moreover, such quantum code is also a $[114,93,9]$ linear code over $\F_4$, which has a {\em better minimum distance} than the current best known linear code over $\F_4$ with the same length and dimension.
\end{example}  


\section*{Acknowledgment}
This work was supported by 
the Natural Sciences and Engineering Research Council of Canada (NSERC), Project No.\ RGPIN-2015-06250 and RGPIN-2022-04526.

\bibliographystyle{abbrv}
\bibliography{DastbastehLisonek}

\end{document}